\documentclass[reprint,aps,pra,groupedaddress,showpacs]{revtex4-1}

\usepackage{amsthm}
\usepackage{amsmath}
\usepackage{amssymb}
\usepackage{mathrsfs}
\usepackage{xspace}

\DeclareMathOperator{\Tr}{Tr}

\newcommand{\QSL}{QSL\xspace}
\newcommand{\QSLs}{QSLs\xspace}
\newcommand{\limitingp}{\pi/2}

\newcommand{\aveE}[1]{\langle |E|^{#1}\rangle}
\newcommand{\dE}[1]{{\mathscr D}_{#1}E}
\newcommand{\dmu}[1]{d_{#1,\vec{\mu}}}
\newcommand{\dmuopt}[1]{d^\triangledown_{#1,\vec{\mu}}}
\newcommand{\desmu}[1]{\mu^\downarrow_{#1}}
\newcommand{\absarg}[2]{|\theta|^\downarrow_{#2}(#1)}
\newcommand{\dg}{d_g}
\newcommand{\dpg}{d^\triangledown_g}
\newcommand{\Eeigket}[1]{|E_{#1}\rangle}
\newcommand{\Eeigbarebra}[1]{\langle E_{#1}}
\newcommand{\Eeigbra}[1]{\Eeigbarebra{#1}|}

\newtheorem{lemma}{Lemma}
\newtheorem{proposition}{Proposition}
\newtheorem{corollary}{Corollary}
\newtheorem{remark}{Remark}

\begin{document}

\title{Relation Between Quantum Speed Limits And Metrics On
 $\boldsymbol{U(n)}$}

\author{Kai-Yan Lee} \altaffiliation[Present address: ]{Department of Astronomy
 and Oskar Klein Centre for Cosmoparticle Physics, Stockholm University,
 Albanova, SE-10691 Stockholm, Sweden} \email{lee.kai\_yan@astro.su.se}
\author{H.\ F. Chau}\email[Corresponding author; ]{hfchau@hku.hk}
 \affiliation{Department of Physics and Center of Computational and Theoretical Physics\\
 University of Hong Kong, Pokfulam Road, Hong Kong}
 
\date{\today}

\begin{abstract}
 Recently, Chau [Quant. Inform. \& Comp. \textbf{11}, 721 (2011)] found a
 family of metrics and pseudo-metrics on $n$-dimensional unitary operators that
 can be interpreted as the minimum resources (given by certain tight quantum
 speed limit bounds) needed to transform one unitary operator to another.
 This result is closely related to the weighted $\ell^1$-norm on
 ${\mathbb R}^n$.
 Here we generalize this finding by showing that every weighted $\ell^p$-norm
 on ${\mathbb R}^n$ with $1\le p \le \limitingp$ induces a metric and a
 pseudo-metric on $n$-dimensional unitary operators with quantum
 information-theoretic meanings related to certain tight quantum speed limit
 bounds.
 Besides, we investigate how far the correspondence between the existence of
 metrics and pseudo-metrics of this type and the quantum speed limits can go.
\end{abstract}


\maketitle

\section{Introduction}\label{Sec:intro}

 Distinguishing two quantum operations and characterizing the resources needed
 to carry out a quantum operation are two meaningful problems in quantum
 information science.
 Various authors have studied the former problem.
 For instance, the problem of unambiguously distinguishing two quantum
 operators have been extensively studied~\cite{WY06,Chefles07,DFY07}.
 Whereas one way to attack the latter problem is through the so-called quantum
 speed limits (\QSLs) which put lower bounds on the evolution time needed to
 perform a unitary operation~\cite{PHYSCOMP96,ML98,GLM03,Chau10}.

 Given a Hamiltonian and an initial state, the evolution time $\tau$ needed to
 perform a unitary operation generated by the Hamiltonian is fixed.
 However, no explicit analytical expression for $\tau$ is known to date.
 The study of QSL makes use of a simple compromise to the above problem by
 asking what $\tau$ could be if a partial description of the quantum system,
 such as the energy standard deviation~\cite{Bh83}, the energy of the system
 above its ground state~\cite{PHYSCOMP96,ML98,GLM03} and the average absolute
 deviation from the median of the energy of the state~\cite{Chau10}, is given.
 Surely, the partial information given is not sufficient to deduce $\tau$.
 Yet surprisingly, non-trivial constraints in the form of explicit evolution
 time lower bounds (called \QSL bounds) can be deduced.  
 Moreover, these bounds are tight in the sense that for each of the above \QSL
 bound, we can find an initial state and a Hamiltonian generating the unitary
 operation such that the required evolution time is equal to the
 lower bound~\cite{Bh83,PHYSCOMP96,ML98,GLM03,Chau10}.
 Interestingly, these bounds are mutually complimentary in the sense that none
 of them always gives a better evolution time lower bound than the
 others~\cite{Chau10}.
 And this is not unexpected for each of these \QSL bounds are deduced using
 different partial information describing the quantum system.

 Note that a few \QSLs have geometric meanings.
 For instance, the well-known time-energy uncertainty relation~\cite{Bh83}
 comes from the Bures metric on the group of unitary operators~\cite{Uh92}.
 And the recently discovered families of metrics and pseudo-metrics on the
 group of $n$-dimensional unitary matrices $U(n)$ by Chau~\cite{Chau11} are
 closely related to a \QSL involving the average absolute deviation from the
 median energy.
 Actually, for any $U,V\in U(n)$, these metrics and pseudo-metrics can be
 written as certain weighted sums of the absolute value of the argument of
 eigenvalue of the unitary matrix $U V^{-1}$; hence, they are related to
 certain weighted average of the absolute value of the energy eigenvalues of
 the Hermitian operator generating $U V^{-1}$.
 Lately, Chau \emph{et al.} went further to show that these families of metrics
 and pseudo-metrics can be induced by the symmetric weighted $\ell^1$-norm on
 ${\mathbb R}^n$~\cite{CLPS12}.
 (We will define symmetric weighted $\ell^p$-norm for $p\ge 1$ in
 Sec.~\ref{Sec:Norm}.)
 More importantly, they~\cite{CLPS12} interpreted these metrics and
 pseudo-metrics as the consequence of certain ``reasonable'' cost functions
 to implement a unitary operation given by the tight \QSL bound reported in
 Ref.~\cite{Chau10}.
 (We will clarify what we mean by a ``reasonable'' cost function in
 Sec.~\ref{Sec:Connection}.)

 It is instructive to study how close the relation between the implementation
 cost in terms of, say, certain \QSLs and the existence of metrics or
 pseudo-metrics on $U(n)$ induced by such cost functions.
 Here we extend the findings by Chau and his
 co-workers~\cite{Chau10,Chau11,CLPS12} by proving the following results.
 First, for $p \ge 1$, there are metrics and pseudo-metrics on $U(n)$ that are
 functions of $|\theta_j|^p$'s where $e^{i\theta_j}$'s are the eigenvalues of
 the unitary matrix $U V^{-1}$ with $\theta_j \in (-\pi,\pi]$ for all $j$.
 In fact, these metrics and pseudo-metrics are induced by certain symmetric
 weighted $\ell^p$-norms on ${\mathbb R}^n$.
 Second, for every $p > 0$, there are two \QSL bounds.
 One involves $\aveE{p}^{1/p}$, the $p$th root of the $p$th moment of the
 absolute value of energy of the system; and the other involves $\dE{p}$, which
 is an optimized version of $\aveE{p}^{1/p}$ by exploiting the freedom of
 choosing the reference energy level.
 Most importantly, these bounds are tight for $p \le \limitingp$.
 Thus, for $1\le p \le \limitingp$, the metrics and pseudo-metrics reported
 here can be interpreted as the minimum resources needed (through the tight
 \QSL bounds involving $\aveE{p}$ and $\dE{p}$ respectively) to convert one
 unitary operator to another.
 Nevertheless, our findings imply that for $p < 1$ or $p \ge \limitingp$, this
 close relation between the metric / pseudo-metric and the \QSL breaks down
 because either the induced metric / pseudo-metric no longer exits or the \QSL
 is no longer a tight bound.
 This work is a refinement and improvement of the research reported in the
 M.Phil.\ thesis of the first author~\cite{Lee11}.

\section{Metrics And Pseudo-metrics Induced By Weighted
 $\boldsymbol{\ell^p}$-Norms}
\label{Sec:Norm}

 We say that a function $g \colon {\mathbb R}^n \to [0,\infty)$ a
 \textbf{symmetric norm} on ${\mathbb R}^n$ if $g$ is a norm on ${\mathbb R}^n$
 satisfying $g({\mathbf v}) = g({\mathbf v}P)$ for any ${\mathbf v} \in
 {\mathbb R}^n$, and any permutation matrix or diagonal orthogonal
 matrix $P$.

 Recall that for any fixed $p \ge 1$, a weighted $\ell^p$-seminorm on
 ${\mathbb R}^n$ is a function $h \colon {\mathbb R}^n \to [0,\infty)$ in the
 form $h({\mathbf v}) \equiv h(v_1,v_2,\dots,v_n) = ( \sum_{j=1}^n \mu_j
 |v_j|^p )^{1/p}$ for some $\mu_j \ge 0$ for all $j$.
 Surely, a weighted $\ell^p$-seminorm is indeed a seminorm on ${\mathbb R}^n$.

 For any weighted $\ell^p$-seminorm $h$, we may define
\begin{align}
 g({\mathbf v}) &= \max_P h(v_{P(1)},v_{P(2)},\dots,v_{P(n)}) \nonumber \\
 &= \left[ \sum_{j=1}^n \desmu{j} \left( |v|^\downarrow_j \right)^p
  \right]^{\frac{1}{p}} ,
 \label{E:symmetric_norm}
\end{align}
 where the maximum is over all permutations $P$ of $\{1,2,\dots,n\}$.
 Besides, $\desmu{j}$ and $|v|^\downarrow_j$ denote the $j$th largest number in
 the sequences $(\mu_1,\mu_2,\dots,\mu_n)$ and $(|v_1|,|v_2|,\dots,|v_n|)$,
 respectively.
 It is straightforward to check that $g$ is a symmetric norm on
 ${\mathbb R}^n$ provided that not all $\mu_j$'s are $0$; and we call this
 particular type of symmetric norm the \textbf{symmetric weighted
 $\boldsymbol{\ell^p}$-norm}.
 (By taking the limit $p \to +\infty$, we have a symmetric weighted
 $\ell^\infty$-norm.
 This symmetric weighted $\ell^\infty$-norm is a special case of symmetric
 weighted $\ell^1$-norm in which all but one of the weights $\mu_j$ are $0$.
 So, we will not pay particular attention to symmetric weighted
 $\ell^\infty$-norm any further.)

 For any symmetric weighted $\ell^p$-norm on ${\mathbb R}^n$, we may apply the
 following result by Chau \emph{et al.} in Ref.~\cite{CLPS12} to induce a
 metric and a pseudo-metric on $U(n)$:

\begin{proposition}[Chau \emph{et al.}]
 For any given symmetric norm $g: {\mathbb R}^n \to [0,\infty)$, we may define
 a metric $\dg$ and a pseudo-metric $\dpg$ on $U(n)$ by
 \begin{equation}
  \dg(U,V) = g(\absarg{U V^{-1}}{1},\dots,\absarg{U V^{-1}}{n})
  \label{E:d_defined}
 \end{equation}
 and
 \begin{equation}
  \dpg(U,V) = \min_{x \in {\mathbb R}} g(\absarg{e^{i x} U V^{-1}}{1},\dots,
  \absarg{e^{i x} U V^{-1}}{n}) .
  \label{E:d_pseudo_defined}
 \end{equation}
 Here $\absarg{UV^{-1}}{j}$ denotes the $j$th largest number in the sequence
 $( |\theta_1|, |\theta_2|, \dots, |\theta_n| )$ with $e^{i \theta_j}$'s being
 the eigenvalues of $U V^{-1}$ obeying $-\pi < \theta_j \le \pi$.
 \label{Prop:Norm}
\end{proposition}

\begin{corollary} \label{Cor:metric_to_qsl}
 Suppose $p \ge 1$.
 Then,
 \begin{widetext}
 \begin{equation}
  \dmu{p}(U,V) = \left\{ \sum_{j=1}^n \desmu{j} \left[ \absarg{U V^{-1}}{j}
  \right]^p \right\}^{\frac{1}{p}} = \left( \sum_{j=1}^n \mu_j
  \right)^{\frac{1}{p}} \min_{H t \colon \exp ( - i H t / \hbar ) = U V^{-1}}
  \quad \max_{|\phi\rangle\in C(H,\vec{\mu})} \ \left[ \aveE{p}(H,|\phi\rangle)
  \right]^\frac{1}{p} t
  \label{E:interpret_l_p_d}
 \end{equation}
 and
 \begin{align}
  \dmuopt{p}(U,V) &= \min_{x\in {\mathbb R}} \ \dmu{p}(e^{i x}U,V)
   = \min_{x\in {\mathbb R}} \left\{ \sum_{j=1}^n \desmu{j} \left[
   \absarg{e^{i x} U V^{-1}}{j} \right]^p \right\}^{\frac{1}{p}} \nonumber \\
  &= \min_{x\in {\mathbb R}} \left( \sum_{j=1}^n \mu_j \right)^{\frac{1}{p}}
   \min_{H t \colon \exp ( - i H t / \hbar) = e^{i x} U V^{-1}} \quad
   \max_{|\phi\rangle\in C(H,\vec{\mu})} \ \dE{p}(H,|\phi\rangle) \ t
  \label{E:interpret_l_p_dopt}
 \end{align}
 \end{widetext}
 are metric and pseudo-metric on $U(n)$, respectively.
 Here $C(H,\vec{\mu})$ is the set of all (normalized) state kets in the form
 $\sum_{j=1}^n \alpha_j |E_{P(j)}\rangle$, $|E_j\rangle$ is the energy
 eigenstate of $H$ with energy $E_j$, $|\alpha_j|^2 = \mu_j / \sum_k \mu_k$,
 and $P$ is a permutation of $\{1,2,\dots,n\}$.
 Also,
 \begin{equation}
  \aveE{p} \equiv \aveE{p}(H,|\phi\rangle) = \Tr (|H|^p
  |\phi\rangle\langle\phi|) = \langle\phi| |H|^p |\phi\rangle \label{E:Ep_Def}
 \end{equation}
 is the $p$th moment of the absolute value of energy of the system and
 \begin{equation}
  \dE{p} \equiv \dE{p}(H,|\phi\rangle) = \min_{x\in {\mathbb R}} \left[
  \aveE{p} (H- xI,|\phi\rangle) \right]^\frac{1}{p}
  \label{E:DpE_Def}
 \end{equation}
 is the $p$th root of the $p$th moment of the absolute value of energy of the
 system minimized over the reference energy level.
\end{corollary}

\begin{proof}
 Applying Proposition~\ref{Prop:Norm} to the symmetric weighted $\ell^p$-norm
 in Eq.~\eqref{E:symmetric_norm} gives the first equality in
 Eq.~\eqref{E:interpret_l_p_d} as well as the first line of
 Eq.~\eqref{E:interpret_l_p_dopt}.
 More importantly, it implies that $\dmu{p}$ and $\dmuopt{p}$ are metric and
 pseudo-metric, respectively.

 To show the second equality in Eq.~\eqref{E:interpret_l_p_d}, we adopt the
 strategy used in the proof of Theorem~1 in Ref.~\cite{Chau11}.
 Note that $\aveE{p}(H,\sum_j \alpha_j |E_j\rangle) = \sum_j |\alpha_j|^2
 |E_j|^p$.
 So, the R.H.S. of Eq.~\eqref{E:interpret_l_p_d} becomes $\min [ \sum_j
 \desmu{j} (|E t|^\downarrow_j)^p ]^{1/p}$, where $|E|^\downarrow_j$ denotes
 the $j$th largest element in the sequence $(|E_1|,|E_2|, \dots,|E_n|)$.
 Among those $H t$'s that satisfy $\exp ( -i Ht / \hbar) = U V^{-1}$, the one
 that minimizes $[ \sum_j \desmu{j} (|E t|^\downarrow_j)^p ]^{1/p}$ can always
 be picked in such a way that its eigenvalues all lie in $(-\pi,
 \pi]$~\cite{Chau11}.
 Hence, the R.H.S. of Eq.~\eqref{E:interpret_l_p_d} is reduced to $\{ \sum_j
 \desmu{j} [ \absarg{U V^{-1}}{j} ]^p \}^{1/p}$.

 We omit the proof of the last line of Eq.~\eqref{E:interpret_l_p_dopt} for it
 is essentially the same as that of the second inequality in
 Eq.~\eqref{E:interpret_l_p_d}.
\end{proof}

\begin{remark} \label{Rem:metric}
 From the above proof, we know that Eqs.~\eqref{E:interpret_l_p_d}
 and~\eqref{E:interpret_l_p_dopt} hold irrespective of whether $\dmu{p}$ is a
 metric or not.
 We note further that the special cases of $\dmu{1}$ and $\dmuopt{1}$ are the
 metric and pseudo-metric reported originally by Chau in Ref.~\cite{Chau11}.
 Moreover, it is possible to use an elementary method involving Minkowski
 inequality to show that $\dmu{p}$ and $\dmuopt{p}$ are metric and
 pseudo-metric, respectively.
 Details can be found in the master thesis of the first author~\cite{Lee11}.
\end{remark}

\section{Quantum Speed Limits Via $\boldsymbol{\aveE{p}^\frac{1}{p}}$ Or
 $\boldsymbol{\dE{p}}$}
\label{Sec:QSL}

 We extend the proof concept used by Chau in Ref.~\cite{Chau10} to find these
 \QSLs.
 And we begin with the following lemma.

\begin{lemma} \label{Lem:power_inequality}
 Let $0 < p\le 2$.
 Further let $f(x) = (1-\cos x)/x^p$ for $x > 0$ and $f(0) = \lim_{x\to 0^+}
 f(x)$.
 Then $A_p \equiv \sup \left\{ (1-\cos x) / x^p \colon x > 0 \right\}$ exists
 and is equal to $\max_{x\in [0,\pi]} f(x) > 0$.
 In fact the maximum is attained by a unique $x_c \in [0,\pi]$.
 And this unique $x_c$ is a decreasing function of $p$ with $x_c > 0$ for $p <
 2$ and $x_c = 0$ when $p = 2$.
 Thus,
 \begin{equation}
  \cos x \ge 1 - A_p \ |x|^p \label{E:power_inequality}
 \end{equation}
 for all $x\in {\mathbb R}$ with equality hold if and only if $x = 0, \pm x_c$.
\end{lemma}

 Let us talk about the geometric meaning of the lemma before proving it.  The
 lemma means that the curve $y = \cos x$ is always above the curve $y = 1 - A
 |x|^p$ provided that $A$ is a sufficiently large positive number.  Besides,
 $A_p$ is the least possible value of $A$ for this to happen.

\begin{proof}
 Since $p \le 2$, $f(x)$ is well-defined and continuous in $[0,\infty)$.
 Moreover, $f(x) > (1-\cos x) / (x+2\pi)^p = f(x+2\pi)$ for all $x > 0$ and
 $f(x) > f(2\pi-x)$ for $0 \le x < \pi$ because $p > 0$.
 Hence, $A_p = \max \left\{ f(x) \colon x\in [0,\pi] \right\} \ge f(\pi) > 0$
 and the maximum is attained by a certain $x_c \in [0,\pi]$.

 Surely, $\left. df/dx \right|_{x=x_c} = 0$ which can be simplified to
 \begin{equation}
  p\ \tan \frac{x_c}{2} = x_c . \label{E:x_c_condition}
 \end{equation}
 Note that the slope of the curve $y = \tan (x/2)$ is strictly increasing for
 $x\in [0,\pi)$ and is equal to $1/2$ at $x = 0$.
 Hence, for $p = 2$, the only solution of Eq.~\eqref{E:x_c_condition} in the
 domain $[0,\pi]$ is $x_c = 0$.
 Whereas for $p\in (0,2)$, $f(0) = 0$.
 So, $x_c > 0$ as $A_p > 0$.
 Now consider the continuous function $\tilde{f}(x) = \tan (x/2) / x$ in
 $(0,\pi)$ with $\tilde{f}[(0,\pi)] = (1/2,\infty)$.
 This function is strictly increasing in $(0,\pi)$ for $d\tilde{f}(x)/dx = (x
 - \sin x ) \sec^2 (x/2)  / 2x^2 > 0$ for $0 < x < \pi$.
 Thus, the equation $\tilde{f}(x) = 1/p$ has a unique solution in the domain
 $(0,\pi)$ for all $p\in (0,2)$.
 Clearly, this unique solution is the required $x_c$ that maximizes $f(x)$.
 More importantly, since $\tilde{f}$ is strictly increasing in $(0,\pi)$,
 $x_c$ decreases as $p$ increases.

 Since the L.H.S. and R.H.S. of Eq.~\eqref{E:power_inequality} are even
 functions, we only need to prove its validity for $x \ge 0$.
 The case of $x > 0$ is a consequence of $A_p \ge (1-\cos x) / x^p$ for all $x
 > 0$; whereas the case of $x = 0$ is trivial.
 Finally, the if and only if condition for Eq.~\eqref{E:power_inequality} to be
 an equality follows from the fact that $f(x)$ is maximized by a unique $x =
 x_c$.
\end{proof}

\begin{corollary} \label{Cor:p_qslb}
 Let $p > 0$, and $H = \sum_{j=1}^n E_j \Eeigket{j}\Eeigbra{j}$ be a
 time-independent Hamiltonian acting on an $n$-dimensional Hilbert space.
 Then the time $\tau$ needed to evolve a pure state $|\Phi(0)\rangle$ to
 $|\Phi(\tau)\rangle$ under the action of $H$ is lower-bounded by 
 \begin{equation} \label{E:p_qslb}
  \tau \geq \tau_{c1} \equiv \hbar \left( \frac{1-\sqrt{\epsilon}}{A_p
  \aveE{p}} \right)^{\frac{1}{p}} .
 \end{equation}
 Here $\epsilon = F(|\Phi(0)\rangle,|\Phi(\tau)\rangle) \equiv | \langle\Phi(0)
 |\Phi(\tau)\rangle |^2$ is the fidelity between the two states, and $\aveE{p}$
 is the $p$th moment of the absolute value of energy of the system defined by
 Eq.~\eqref{E:Ep_Def}.
 Also, $A_p$ is given by Lemma~\ref{Lem:power_inequality} if $p\le 2$ and $A_p$
 is defined to be $A_2$ otherwise.
 Actually, we can slightly optimize the bound in Eq.~\eqref{E:p_qslb} to
 \begin{equation} \label{E:p_qslbopt}
  \tau \geq \tau_{c2} \equiv \frac{\hbar}{\dE{p}} \left(
  \frac{1-\sqrt{\epsilon}}{A_p} \right)^{\frac{1}{p}} ,
 \end{equation}
 where $\dE{p}$ is the $p$th root of the $p$th moment of the absolute value of
 energy of the system minimized over the reference energy level as defined by
 Eq.~\eqref{E:DpE_Def}.
 More importantly, these two bounds are tight for all $\epsilon \in [0,1]$ if
 $p \le \limitingp$.
\end{corollary}

\begin{proof}
 We first prove Eq.~\eqref{E:p_qslb} for the case of $p \le 2$.
 The initial quantum state $|\Phi(0)\rangle$ can be written in the form
 $\sum_{j=1}^n \alpha_j \Eeigket{j}$ with $\sum_{j=1}^n |\alpha_j|^2 = 1$.
 From Lemma~\ref{Lem:power_inequality}, the time $\tau > 0$ needed to evolve to
 a state $|\Phi(\tau)\rangle$ with $F(|\Phi(0)\rangle,|\Phi(\tau)\rangle) =
 \epsilon$ must obey
 \begin{widetext}
 \begin{align} \label{E:inequality1}
  \sqrt{\epsilon} &= \left| \langle \Phi (0)|\Phi (\tau)\rangle \right|
   = \left| \sum_{j=1}^n |\alpha_j|^2 e^{-i E_j \tau / \hbar} \right| 
   \ge \left| \sum_{j=1}^n |\alpha_j|^2 \cos \left( \frac{E_j \tau}{\hbar}
    \right) \right|
   \ge \sum_{j=1}^n |\alpha_j|^2 \cos \left( \frac{E_j \tau}{\hbar} \right)
    \nonumber \\
  &\ge \sum_{j=1}^n |\alpha_j|^2 \left( 1 - A_p \left| \frac{E_j \tau}{\hbar}
    \right|^p \right)
   = 1 - \frac{A_p \tau^p}{\hbar^p} \,\sum_{j=1}^n |\alpha_j|^2 |E_j|^p
   = 1 - \frac{A_p \tau^p \aveE{p}}{\hbar^p} .
 \end{align}
 \end{widetext}
 Hence, the \QSL in Eq.~\eqref{E:p_qslb} is valid whenever $0 < p \le 2$.

 To prove the validity of this \QSL for the case of $p > 2$, we only need
 to combine Eq.~\eqref{E:p_qslb} for $p = 2$ and $\aveE{p}^{1/p} \ge
 \aveE{2}^{1/2}$ for all $p > 2$, which is a special case of the Lyapunov's
 inequality~\cite{*[{See, for example, }][{ for a proof.}] Sh95}.

 Now have just established the truth of the \QSL in Eq.~\eqref{E:p_qslb}.
 The other \QSL given by Eq.~\eqref{E:p_qslbopt} follows by the fact that the
 reference energy level has no physical meaning.
 So from Eq.~\eqref{E:p_qslb}, we can obtain a more ``optimized'' bound by
 varying the reference energy level $x$ so as to minimize $\aveE{p}(H - x I,
 |\Phi(0)\rangle)$~\cite{Chau10}.
 Therefore, Eq.~\eqref{E:p_qslbopt} follows from Eqs.~\eqref{E:DpE_Def}
 and~\eqref{E:p_qslb}.

 To show that the two \QSLs are tight bounds for $p \le \limitingp$, we only
 need to give an example of initial state that saturates the bound for all
 fidelity $\epsilon \in [0,1)$.
 And since the bound in Eq.~\eqref{E:p_qslbopt} is in general more stringent
 than the bound in Eq.~\eqref{E:p_qslb}, we only need to show that the example
 we give saturates the former bound.
 Note further that there is no need to check for the case of $\epsilon = 1$
 because the \QSLs reduce to $\tau \ge 0$ which is trivially true.
 Now we claim that the following state saturates the \QSL stated in
 Eq.~\eqref{E:p_qslb}:
 \begin{equation}
  |\Phi(0)\rangle = \sqrt{1-\beta} |0\rangle + \sqrt{\frac{\beta}{2}} \ \left(
  |-{\mathcal E}\rangle + |{\mathcal E} \rangle \right) \label{E:magic_state}
 \end{equation}
 where ${\mathcal E} > 0$ and $\beta = (1-\sqrt{\epsilon})/A_p x_c^p$ with
 $x_c$ being the maximum point defined in Lemma~\ref{Lem:power_inequality} so
 that $\cos x_c = 1 - A_p x_c^p$.
 (Note that $\beta$ is well-defined as Lemma~\ref{Lem:power_inequality} demands
 $A_p > 0$.)
 Since $p \leq \pi/2 < 2$, Lemma~\ref{Lem:power_inequality} implies $x_c > 0$.
 Thus, $\beta \ge 0$ for all $\epsilon \in [0,1]$.
 As $x_c$ is a decreasing function of $p$ obeying Eq.~\eqref{E:x_c_condition},
 $A_p x_c^p = 1 - \cos x_c \ge 1$ (and hence $\beta \le 1$ for all $\epsilon
 \in [0,1]$) whenever $0 < p \le p_c$ where $p_c$ is the critical value of $p$
 in $(0,2]$ such that $\cos x_c = 0$ and hence $x_c = \pi/2$.
 From Eq.~\eqref{E:x_c_condition}, $p_c = \limitingp$.
 In conclusion, Eq.~\eqref{E:magic_state} is a valid quantum state if $0 < p
 \le \limitingp$.
 Since $A_p, x_c > 0$ and $\epsilon < 1$, it is easy to check that for this
 particular quantum state, $\aveE{p} = \beta {\mathcal E}^p$ and
 $\langle\Phi(0)|\Phi(\tau_c)\rangle = 1 - \beta + \beta \cos x_c =
 \sqrt{\epsilon}$.
 So Eq.~\eqref{E:p_qslb} is tight for all $\epsilon \in [0,1]$ provided that
 $p \in (0,\limitingp]$.
\end{proof}

\begin{remark} \label{Rem:p_qslb}
 From the above proof, for $\limitingp < p < 2$, Eqs.~\eqref{E:p_qslb}
 and~\eqref{E:p_qslbopt} are tight for some but not all $\epsilon \in [0,1]$
 because Eq.~\eqref{E:magic_state} is a valid quantum state for $\epsilon$
 sufficiently close to $1$.
 Note further that for $p = 1$, Corollary~\ref{Cor:p_qslb} reduces to an
 earlier result obtained by Chau in Ref.~\cite{Chau10}.
 Actually, the \QSLs reported here also apply to the case of mixed state
 through the use of the purification argument by Giovannetti \emph{et al.} in
 Ref.~\cite{GLM03}.
 Hence, these two \QSLs can be regarded as fundamental limit on the minimum
 time needed to evolve a density matrix or alternatively as a fundamental limit
 on the maximum possible information processing rate of a
 system~\cite{PHYSCOMP96,ML98,GLM03}.
\end{remark}

 Table~\ref{T:compare} shows the actual evolution time $\tau$ and the \QSL
 bounds $\tau_{c2}$ for different values of $p$ for a few selected states in
 which $\tau$ can be calculated exactly.
 The larger the value of $\tau_{c2} / \tau$, the better the estimation is the
 lower bound.
 The table shows that while $\tau_{c2}$'s are different for different choice of
 $p$, in general they all give reasonably good estimates on the actual $\tau$.
 This is true even for the case of $p = 2$, which is not a tight bound.
 To understand why it is so, we start from Lemma~\ref{Lem:power_inequality} and
 Eq.~(\ref{E:p_qslbopt}).
 They imply
\begin{eqnarray}
 \frac{\tau_{c2} \dE{p}}{\hbar} & \leq & A_p^{-1/p} = \inf_{x > 0} \left\{
  \frac{x}{(1-\cos x)^{1/p}} \right\} \nonumber \\
 & \leq & \frac{\pi / 2}{[1-\cos ( \pi / 2)]^{1/p}} = \frac{\pi}{2}
 \label{E:angle_bound_p_qslb}
\end{eqnarray}
 for all $\epsilon \in [0,1]$.
 (Interestingly, similar conclusions can be deduced for other \QSL bounds
 including the time-energy uncertainty bound~\cite{Bh83} and the
 Margolus-Levitin bound~\cite{PHYSCOMP96,ML98}.
 Their proofs are straightforward and are left to the readers.)
 Let us use the notation using in Corollary~\ref{Cor:p_qslb} to express the
 initial state $|\Phi(0)\rangle$ as $\sum_j \alpha_j |E_j\rangle$.
 Provided that $\vartheta_j \equiv \tau |E_j| / \hbar \leq \pi$ for all $j$
 (that is, the phase angle $\vartheta_j$ rotated for each eigen-energy mode in
 the time evolution is at most $\pi$), we have $\tau \dE{p} / \hbar \lesssim
 \pi / 2$.
 Hence, $\tau \approx \tau_{c2}$.
 So, our \QSL bounds $\tau_{c2}$'s generally give reasonably good estimates to
 the actual evolution times $\tau$'s for all the initial states listed in
 Table~\ref{T:compare} because these states are picked so that $\vartheta_j
 \leq \pi$ for all $j$.
 Readers will find in Sec.~\ref{Sec:Connection} that this discussion is
 essential to understand why our \QSL bounds can be used to study the resources
 needed to carry out certain unitary operations.

\begin{widetext}
\begin{table*}[t]
 \centering
 \begin{tabular}{||c|c|c|c|c|c|c||}
  \hline
   initial state & $\tau$ & \multicolumn{5}{c||}{$\tau_{c2} / \tau$} \\
  \cline{3-7}
   & & $p = 0.1$ & $p = 0.5$ & $p = 1.0$ & $p = 1.5$ & $p = 2.0$ \\
  \hline
  $\frac{1}{\sqrt{2}} \left( |-{\mathcal E} \rangle + |{\mathcal E}\rangle
   \right)$ & $\frac{\pi\hbar}{2{\mathcal E}}$ & $0.9897$ & $0.9450$ & $0.8786$
   & $0.9982$ & $0.9003$ \\
  \hline
  $|\Phi(0)\rangle$ ~as defined in~Eq.~(\ref{E:magic_state}) for $\beta = 1 /
   A_1 x_c$ & $\frac{x_c \hbar}{{\mathcal E}}$ & $0.2463$ & $0.9084$ & $1.0000$
   & $0.9540$ & $0.7885$ \\
  \hline
  $|\Phi(0)\rangle$ ~as defined in~Eq.~(\ref{E:magic_state}) for $\beta =
   \frac{4}{4-\sqrt{2}+\sqrt{6}} \approx 0.794$ &
   $\frac{7\pi\hbar}{12 {\mathcal E}}$ & $0.0167$ & $0.6879$ & $0.9480$ &
   $0.9975$ & $0.8658$ \\
  \hline
  $|\Phi(0)\rangle$ ~as defined in~Eq.~(\ref{E:magic_state}) for $\beta = 1/2$
   & $\frac{\pi\hbar}{{\mathcal E}}$ & $0.9897$ & $0.9450$ & $0.8786$ &
   $0.7923$ & $0.6366$ \\
  \hline
  $\frac{1}{\sqrt{3}} \left( |0\rangle + |-{\mathcal E}\rangle +
   |{\mathcal E}\rangle \right)$ & $\frac{2\pi\hbar}{3{\mathcal E}}$ & $0.0836$
   & $0.7973$ & $0.9884$ & $0.9810$ & $0.8270$ \\
  \hline
  $\frac{1}{\sqrt{2n+1}} \sum_{k=-n}^n |k {\mathcal E} \rangle$ in the large
   $n$ limit & $\frac{2\pi\hbar}{(2n+1) {\mathcal E}}$ & $0.0025$ & $0.5316$ &
   $0.8786$ & $0.9194$ & $0.7797$ \\
  \hline
 \end{tabular}
 \caption{Comparison between $\tau$ and the lower bound $\tau_{c2}$ for
  $\epsilon = 0$ on a few states using different values of $p$ based on a
  similar table in Ref.~\cite{Chau10}.
  \label{T:compare}
 }
\end{table*}
\end{widetext}

\section{Connection Between The Metrics, Pseudo-Metrics And The Quantum Speed
 Limit Bounds}
\label{Sec:Connection}

 By comparing $\dmu{p}$ in Eqs.~\eqref{E:interpret_l_p_d} and $\dmuopt{p}$ in
 Eq.~\eqref{E:interpret_l_p_dopt} of Corollary~\ref{Cor:metric_to_qsl} with the
 \QSLs involving $\aveE{p}$ or $\dE{p}$ in Corollary~\ref{Cor:p_qslb}, we may
 interpret $\dmu{p}$ and $\dmuopt{p}$ as cost functions describing the minimum
 amount of resources needed to convert $V$ from $U$.
 In the first case, the resources refer to the product of evolution time $\tau$
 and the $p$th root of the $p$th moment of absolute value of energy of the
 system $\aveE{p}^{1/p}$ required to carry out the conversion.
 And in the second case, the resources refer to the product of $\tau$ and
 $\dE{p}$ (which is an ``optimized'' version of
 $\aveE{p}^{1/p}$)~\cite{Chau11}.
 Note that the out of the Hamiltonians that generate a given unitary operation,
 we can always pick one so that $\vartheta_j \leq \pi$ for all $j$.
 Thus, the discussion in the final paragraph of Sec.~\ref{Sec:QSL} implies that
 our cost functions are reasonably good estimates of the actual amount of
 resources required to covert $U$ from $V$.

 Three remarks are in place.
 First, since this connection is done via \QSL bounds, it works best when the
 bounds are tight for all $\epsilon$.
 For otherwise, the cost functions always overestimate the actual resources
 required.
 So, from Corollary~\ref{Cor:p_qslb}, this connection begins to lose its
 significance when $p > \pi / 2$.

 Second, from Remark~\ref{Rem:metric}, we know that this interpretation works
 whenever $p > 0$ --- that is, even in the case when $\dmu{p}$ is not a metric.
 However, Chau \emph{et al.}~\cite{CLPS12} argued that any ``reasonable'' cost
 functions $\dmu{p}$ and $\dmuopt{p}$ should be a metric and a pseudo-metric on
 $U(n)$, respectively.
 Part of the reasons is that one way to transform $V$ to $U$ is first
 transforming $V$ to $W$ and then from $W$ to $U$.
 So, $\dmu{p}$ and $\dmuopt{p}$ must satisfy the triangle inequality if the
 cost of transformation is additive --- a rather modest additional requirement
 indeed.
 In this regard, the cost functions $\dmu{p}$ and $\dmuopt{p}$ in
 Eqs.~\eqref{E:interpret_l_p_d} and~\eqref{E:interpret_l_p_dopt} are
 ``reasonable'' provided that $p \ge 1$.

 Finally, since the overall phase of a unitary operator has no physical
 significance, the cost function $\dmuopt{p}(U,V)$ is more meaningful than
 $\dmu{p}(U,V)$ in characterizing the resources needed to transform $V$ to $U$.
 Nonetheless, $\dmu{p}$ is important in its own right for it gives rise to a
 characterization on the degree of non-commutativity between two unitary
 operators $U$ and $V$ through $\dmu{p}(U V,V U)$~\cite{Chau11}.

 To summarize, we have shown that any symmetric weighted $\ell^p$-norm on
 ${\mathbb R}^n$ induces a metric and a pseudo-metric on $U(n)$ for $p \ge 1$.
 These metrics and pseudo-metric can be interpreted as ``reasonable'' cost
 functions described by tight \QSL bounds involving $\aveE{p}$ and $\dE{p}$
 respectively provided that $p \in [1,\pi/2]$.
 There is an open problem along this line of study.
 Our numerical study strongly suggests that $\dmu{p}(U,V) = \sum_{j=1}^n
 \desmu{j} [ \absarg{U V^{-1}}{j} ]^p$ is a metric on $U(n)$ for $0 < p < 1$
 whenever $\desmu{1} > 0$.
 It is instructive to prove this conjecture and to relate it to a tight \QSL
 bound.

\begin{acknowledgments}
 We would like to thank C.-H.\ F.\ Fung for his valuable discussions.
 This work is supported under the RGC grant HKU~700709P of the HKSAR
 Government.
\end{acknowledgments}

\bibliography{qc55.2}
\end{document}